\documentclass{IEEEtran}
\usepackage{amsmath,amsthm,amssymb}
\usepackage{tikz}
\def\|{\,|\,}
\def\restr{\mathord\upharpoonright}
\renewcommand\ge{\geqslant}\renewcommand\le{\leqslant}

\newcommand\B{\mathcal B}
\newcommand\C{\mathcal C}
\newcommand\D{\mathcal D}
\newcommand\I{\mathcal I}
\newcommand\J{\mathcal J}

\newcommand\N{\mathbb N}
\newcommand\Sn{\mathcal S_n}
\renewcommand\phi{\varphi}
\DeclareMathAlphabet{\mathbfit}{OML}{cmm}{b}{it}
\def\barH{\hbox to 0pt{\kern 0.5ex$\overline{\rule[1.8ex]{1.8ex}{0pt}}$\hss}\mathbfit H}
\def\Hen{\mathbfit H}
\newcommand\eqdef{\stackrel{\mathrm{def}}{=}}
\newcommand\dn{\mathord{\downarrow}}
\newcommand\tsum{{\textstyle\sum}}
\newcommand\cups{\mathop{{\cup}^*}}

\newtheorem{theorem}{Theorem}
\newtheorem{definition}[theorem]{Definition}
\newtheorem{proposition}[theorem]{Proposition}
\newtheorem{claim}[theorem]{Claim}
\newtheorem{lemma}[theorem]{Lemma}
\newtheorem{conjecture}[theorem]{Conjecture}
\begin{document}
\title{Book inequalities}
\author{Laszlo Csirmaz%
\thanks{%
Central European University, Budapest, Hungary \par
 e-mail: {\tt csirmaz@renyi.hu}
\par\noindent
Research partially supported by TAMOP-4.2.2.C-11/1/KONV-2012-0001 and 
by the Lendulet program of the Hungarian Academy of Sciences}
}
\date{}
\maketitle
\begin{abstract}
Information theoretical inequalities have strong ties with polymatroids and
their representability. A polymatroid is entropic if its rank function is
given by the Shannon entropy of the subsets of some discrete random
variables. The {\em book} is a special iterated adhesive extension of a
polymatroid with the property that entropic polymatroids have $n$-page book
extensions over an arbitrary spine. We prove that every polymatroid has an
$n$-page book extension over a single element and over an all-but-one-element 
spine. Consequently, for polymatroids on four elements, only book 
extensions over a two-element spine should be considered.
{F.} Mat\'u\v s proved that the
Zhang-Yeung inequalities characterize polymatroids on four elements which
have such a 2-page book extension. The {\em $n$-page book inequalities},
defined in this paper, are conjectured to characterize polymatroids on four
elements which have $n$-page book extensions over a two-element spine. We
prove that the condition is necessary; consequently every book inequality
is an information inequality on four random variables. Using computer-aided
multiobjective optimization, the sufficiency of the condition is verified up
to 9-page book extensions.


\noindent{\bf Classification numbers:} 05B35, 26A12, 52B12, 90C29, 94A17

\end{abstract}

\section{Introduction}\label{sec:intro}

The entropy function of $N$ random variables $\langle x_i:i\in
N\rangle$ maps the non-empty subsets $I\subseteq N$ to the Shannon entropy
$H(\xi_I)$ of the variable set $\xi_I=\langle x_i:i\in I\rangle$. 
The range of the entropy
function, a subset of the $2^N-1$-dimensional Euclidean space,
is denoted by $\Hen_N$. The closure 
$\barH_N$ (in the usual Euclidean topology) of
$\Hen_N$ is a closed, convex, pointed
cone, and $\Hen_N$ misses only some boundary points as shown in \cite{matus}.

The region $\barH_N$ is bounded by linear facets corresponding to the
Shannon entropy inequalities.
Hyperplanes cutting into the Shannon polyhedron and
containing all entropic points on one side are the {\em non-Shannon linear
information inequalities}. The first such inequality was found by Zhang and
Yeung \cite{zhang-yeung}. Later the list of such inequalities has been
extended significantly, see \cite{dougherty-etal,matus-infinite,xu-wang-sun}.
The method of Zhang and Yeung motivated the definition of {\em adhesive
extensions} of polymatroids by {F.} Mat\'u\v s in \cite{matus-adhesive}.
An alternative technique for generating non-Shannon inequalities
was suggested by {K.} Makarychev {\em et al}
\cite{MMRV}, which later was found to rely on the same extension property 
of entropic polymatroids \cite{kaced, matus-csirmaz}.

Section \ref{sec:polymatroids} recalls some notation and terminology related
to
polymatroids; for a detailed account, see \cite{lovasz}. Section \ref{sec:book-extension}
describes the {\em book}, a special iterated adhesive extension. 
Generalizing results from \cite{matus-adhesive} and \cite{matus-csirmaz}, we
prove that book extensions always exist when the spine of the book has one
element, or has all but one elements of the ground set. Sections
\ref{sec:conjecture} and \ref{sec:book-proof} concentrate on the case $N=4$. 
Section \ref{sec:conjecture} defines the collection of {\em book
inequalities}, which is conjectured to characterize polymatroids on four
elements which have $n$-page book extensions. In Section
\ref{sec:book-proof} we prove the necessary part of the conjecture, that is,
that book inequalities hold for polymatroids with $n$-page book extensions. 
As entropic polymatroids have this extension
property, book inequalities are, consequently, entropy inequalities. The book
inequalities contain, among others, one of the infinite lists of Mat\'u\v s
in \cite{matus-infinite}, the list of Dougherty et al \cite[Theorem
10]{dougherty-etal}, and provide infinitely many new information inequalities.
The sufficiency part of the conjecture is left as an open problem.

The collection of book inequalities along with the conjecture that they
characterize the book extensions were announced at the First Workshop on
Entropy and Information Inequalities held in Hong Kong, April 15--17, 2013.
After the conference {\em Randall Dougherty} (personal communication)
pointed out a misprint in the formulation of inequalities in
(\ref{eq:book2}), and supplied a proof for the correct version. In this
paper an alternate proof of his result is given along the way inequalities
in (\ref{eq:book1}) are proved.


\section{Definitions and notation}\label{sec:polymatroids}
Let $N$ be a finite set, and $g$ be a real-valued function on the non-empty
subsets of $N$. The pair $\langle g,N\rangle$ is a {\em polymatroid} if $g$
is non-negative and non-decreasing: that is, $0\le g(I)\le g(J)$ for $I\subseteq 
J\subseteq N$; and submodular:
$$
    g(I)+g(J)-g(I\cup J)-g(I\cap J) \ge 0, ~~~ I,J\subseteq N .
$$
Here $N$ is the {\em ground set}, and $g$ is the {\em rank function}. Polymatroids
and their rank functions are frequently identified. Shannon inequalities for
discrete random variables express the fact that an entropy function is 
a polymatroid. Polymatroids coming from
entropy functions are called {\em entropic}, and those in the closure of
entropic polymatroids are {\em almost entropic}.

\medskip
For $I\subseteq N$ let $\delta_I$ be the
$2^N-1$-dimensional unit vector whose $I$-coordinate is equal to 1, and all
other coordinates are 0. Writing
$$
  (I,J)\eqdef \delta_I + \delta_J - \delta_{I\cup J}-\delta_{I\cap J},
$$
the expression $(I,J)\cdot g$ can be interpreted as the scalar product of
$(I,J)$ with $g$, thus the submodularity of $g$ can be expressed as
$$
    (I,J)\cdot g \ge 0, ~\qquad ~I, J\subseteq N .
$$
We will also use other abbreviations for certain information theoretic 
expressions:
\begin{align*}
(I,J\|K) &\eqdef \delta_{I\cup K}+\delta_{J\cup K}-
\delta_{I\cup J\cup K}-\delta_{(I\cap J)\cup K},\\
{}[I,J,K,L\,] &\eqdef -(I,J)+(I,J\|K)+(I,J\|L)+(K,L).
\end{align*}

For any polymatroid $g$, 
$(I,J\|K)\cdot g\allowbreak\ge 0$ follows from submodularity and
monotonicity.
$[I,J,K,L\,]\cdot g\ge 0$ is the so-called {\em Ingleton inequality}
\cite{ingleton}, and it holds when $g$ is linearly representable over a
field, but not necessarily holds when $g$ is only (almost) entropic.

Following the usual practice, the union symbol is omitted as well as the
curly brackets around singletons. Thus, for example, $aI$ denotes the set
$\{a\}\cup I$, and
$$
   \big( [abcd\,]+(a,b\|c)+(a,c\|b)+(b,c\|a)\big)\cdot g \ge 0
$$
is an equivalent form of the
Zhang-Yeung inequality \cite{zhang-yeung} on  the four-element set 
$N=\allowbreak\{a,b,c,d\}$.
Additionally, we omit the commas in the Ingleton notation $[a,b,c,d\,]$ as
we did it above, and even the polymatroid $g$ is omitted
when it is clear from the context
which polymatroid we are referring to.

The symbol $\cups$ is used to emphasize that the sets whose union is taken
are disjoint.

\subsection{Operations on polymatroids}
This section recalls some basic operations on polymatroids and their
properties.

\subsubsection{Direct sum}
The {\em direct sum} of polymatroids $\langle g_i,N_i\rangle$ for
$i=1,\dots,n$ is the
polymatroid $\langle g,N\rangle$ where the ground set $N$ is the disjoint
union $N_1\cups\dots\cups N_n$, and for every $I_i\subseteq N_i$,
$i=1,\dots,n$ the value of $g$ is defined as
$$
     g(I_1\cups\dots\cups I_n) = g_1(I_1)+\cdots+g_k(I_n).
$$
\subsubsection{Independence}
Let $\langle g,N\rangle$ be a polymatroid, and $P_1$, $P_2$, $S$ be disjoint 
subsets of the ground set $N$. 
$P_1$ and $P_2$ are {\em independent over $S$} when
$(P_1,P_2\|S)=0$, that is, when
$$
   g(P_1S)+g(P_2S) - g(P_1P_2S)-g(S) = 0.
$$
In matroid terminology, $\langle P_1S,P_2S\rangle$ is a {\em modular pair of
$g$}. Let $P_1,\dots,P_n$ and $S$ be disjoint subsets of $N$. The $P_i$s
are {\em totally independent over $S$} if for any two disjoint subsets
$\{i_1,i_2,\dots,i_{n_1}\}$ and $\{j_1,j_2,\dots,j_{n_2}\}$ 
of the indices $1,2,\dots,n$
\begin{equation}\label{eq:independent}
( P_{i_1}P_{i_1}\dots P_{i_{n_1}},\, P_{j_1}P_{j_2}\dots P_{j_{n_2}}\|S)=0 .
\end{equation}
In this case the collection $\{P_1S,P_2S,\dots,P_nS\}$ is called a
{\em modular set}. 
We will use the notation $\langle i \rangle$ to denote the set 
$\{i_1,i_2\dots,i_{n_1}\}$, and $P_{\langle i\rangle}$ to denote the
disjoint union $\bigcup\{ P_i: i\in\langle i\rangle\}$.
Condition (\ref{eq:independent}) can be
written more succinctly as
$$
  (  P_{\langle i\rangle},P_{\langle j\rangle} \,|S ) =0 
$$
for disjoint subsets $\langle i\rangle$, $\langle j\rangle$ of
$\{1,2,\dots,n\}$.

\subsubsection{Restriction}
Restricting the rank function of the polymatroid $\langle g,N\rangle$ to the
subsets of $M\subseteq N$ gives the polymatroid $g\restr M$, the {\em
restriction of $g$ to $M$}; furthermore, $g$ is the {\em extension} of its
restrictions. Restricting an (almost) entropic polymatroid gives an (almost)
entropic polymatroid.

\subsubsection{Pullback}
Let $\phi$ map $N'$ into $N$, and let $g$ be a polymatroid on $N$.
The {\em pullback $\phi^{-1}g$} is the polymatroid defined on the ground 
set $N'$ by
$$
    (\phi^{-1}g)(I') = g(\phi(I')) ~~~~ \mbox{for all $I'\subseteq N'$}.
$$
Thus, for example, restricting $g$ to $M\subseteq N$ is the same as the
pullback $\mathrm{Id}_M^{-1}g$, where $\mathrm{Id}_M$ is the identity map on
$M$. Again, the pullback of an (almost) entropic polymatroid is
(almost) entropic.

\subsection{A technical lemma}
The following lemma describes a polymatroid construction. It will 
be used in the proof of the main result in Section \ref{sec:book-extension}.
\begin{lemma}\label{lemma:minimum-poly}
Let $\langle g,N\rangle$ be a polymatroid, $a\in N$, and $t\le g(a)$. Define
the function $h$ on the non-empty subsets $J\subseteq N$ as follows:
$$
  h(J)  = \min \,\{\,g(J),g(aJ)-t\,\} .
$$
Then, $\langle h,N\rangle$ is a polymatroid.
\end{lemma}
\begin{proof}
The condition $t\le g(a)$ gives $g(aJ)=g(aJ)-t\ge 0$, thus $h$ is
non-negative. As the monotonicity of $h$ is clear, only the
subadditivity needs to be checked. Distinguishing four cases depending on
where the minimum is taken in $h(I)$ and $h(J)$, in each case 
the submodularity of $g$ entails
that their sum is at least as large as $h(I\cup J)+h(I\cap J)$.
\end{proof}

\subsection{Tightening}
Let $\langle g,N\rangle$ be a polymatroid, and $a\in N$. The polymatroid
$\langle g\dn a,N\rangle$ is defined as follows. For each $I\subseteq
N-\{a\}$,
\begin{align*}
   (g\dn a)(I)  &= g(I), \\
   (g\dn a)(aI) &= g(aI)-\big(g(N)-g(N-a)\big) . 
\end{align*}
Applying Lemma \ref{lemma:minimum-poly} with $t=g(N)-g(N-a)$, and observing
that $g(J)\le g(aJ)-t$ by submodularity of $g$, we see that
$g\dn a$ is indeed a polymatroid on $N$. Moreover, this
operation is idempotent: $(g\dn a)\dn a = g\dn a$, and commutative:
$(g\dn a)\dn b = (g\dn b)\dn a$. For subsets $J\subseteq N$ we
define $g\dn J$ as follows. If $J=\{a_1,\dots,a_k\}$, then we let
$$
 g\dn J = (\cdots((g\dn a_1)\dn a_2)\dn \cdots ) \dn a_k.
$$ 
By commutativity, the result depends only on the subset $J$ and not on 
the order of its elements. As $g=g\dn a$ if and only if
$g(N)=g(N-\{a\})$, it follows that $g=g\dn N$ if and only if every 
co-singleton has full rank. Such polymatroids are called {\em tight} in
\cite{matus-csirmaz}.


\section{Book extension}\label{sec:book-extension}

The notion of {\em adhesive extension}, introduced by {F.} Mat\'u\v s in
\cite{matus-adhesive}, captures the essence of the Zhang-Yeung method which
can be outlined as follows. Suppose that the rank function $g$ is given
by the Shannon entropy of the subsets of the random variables $\langle \vec
x,\vec s\rangle$. Using the terminology of Dougherty et al
\cite{dougherty-etal}, the collection of random variables $\vec y$ is a {\em
copy of $\vec x$ over $\vec s$} if $\vec x$ and $\vec y$ are independent
over $\vec s$; otherwise, $\langle \vec x,\vec s\rangle$ and $\langle\vec
y,\vec s\rangle$ have the same distribution. The polymatroid $h$ defined by
the entropies of the (subsets of the) random variables $\langle\vec x,\vec
y,\vec s\rangle$ extends $g$ in two different ways: $g$ can be embedded as
$\langle \vec x,\vec s\rangle$ or as $\langle \vec y,\vec s\rangle$, and
these instances of $g$ form a modular pair in $h$. Polymatroids with this
special embeddability property are called {\em self-adhesive at $\vec s$} in
\cite{matus-adhesive}.

In the above process we could add several independent copies of $\vec x$
instead of adding just a single copy. 
The {\em book extension} generalizes Mat\'u\v s' notion of
adhesivity along this line. 
This generalization, however, does not increase the
strength of the iterated method as $n$ consecutive copy steps over the same
set of variables give $2^n$ many totally independent copies of the pasted
variables.  

\begin{definition}[Book extension]\label{def:book}
Let $\langle g,P\cups S\rangle$ be a polymatroid.
$\langle h,M\rangle$ is an {\em $n$-page book extension of 
$g$ over $S$}, if the ground set of $h$ is the disjoint union
$M = P_1\cups\dots\cups P_n\cups S$ such that

\noindent
\hbox to1.5\parindent{\hss\textup{(i) }}%
$P_1,\dots,P_n$ are totally independent over $S$;

\noindent\hangindent1.5\parindent\hangafter1
\hbox to1.5\parindent{\hss\textup{(ii) }}%
for $i=1,2,\dots,n$ there are bijections $\phi_i:P\cups
S\leftrightarrow P_i\cups S$ which are identity on $S$ and
the pullback of $h$ along $\phi_i$ is $g$: $g=\phi_i^{-1}h$.

\noindent
We write $g\prec^n_S h$ to denote that $h$ is an $n$-page book extension of
$g$ over $S$.
\end{definition}

We use the picturesque name {\em book} for such an extension $h$. $S$ is
the {\em spine} of the book, and the $P_i$'s are its {\em pages}. 
A 2-page extension with spine $S$ is the same as the adhesive extension at
$S$ in \cite{matus-adhesive}. This book is not too interesting as all of its 
pages are the same, the interesting features come from the interaction
between the pages.

A book extension over an empty spine is the same as the direct sum, and when
$S$ is the full ground set, then there is no condition to satisfy. Moreover, as
every polymatroid is a
1-page book extension of itself,
we always assume that $n\ge 2$, and
the spine $S$ is a proper, non-empty subset of the ground set of $g$.
The following properties of the book extension follow immediately from the 
definition. 

\begin{proposition}\label{thm:book-basic}
{\rm a)} If $g\prec^k_S h$ and $h \prec^\ell_S h'$, then $g\prec^{k\ell}_S h'$.
{\rm b)} If $g\prec^n_S h$, and $h'$ is $h$ restricted to $S$ and 
$k$ of its pages, then $g\prec^k_S h'$. In particular, if $g$ has an
$n$-page book extension, then it has $\ell$-page extensions for every
$\ell< n$.
\qed
\end{proposition}

Let $\langle h,M\rangle$ be an $n$-page extension of $\langle g,N\rangle$
over $S$ with bijection $\phi_i$ between $PS=P\cup S$ and $P_iS$. 
Any permutation $\pi$ of the page indices $\{1,2,\dots,n\}$ 
determines a permutation $\sigma_\pi$ of the ground set $M$ by keeping 
$S$ fixed, and by permuting the pages according to $\pi$:
$$
    \sigma_\pi(a)=\begin{cases}
      a & \text{if $a\in S$,}\\
      \phi_{\pi(i)}\phi^{-1}_i(a) & \text{if $a\in P_i$}.
\end{cases}
$$
Subsets $I$ and $J$ of $M$ are called {\em symmetrical} if $\sigma_\pi(I)=J$ 
for some permutation $\pi$ of the pages. This happens if and only if the 
following two conditions hold: $I$ and $J$ intersect the spine in the same 
set: $I\cap S=J\cap S$; and the $n$-element multisets $\{ \phi_i^{-1}(P_i\cap I) \}$ and $\{
\phi_i^{-1}(P_i\cap J)\}$, which consist of subsets of $P$ with multiplicity,
are the same. We call the extension $h$ {\em symmetrical} if symmetrical 
subsets have the same $h$-value.
\begin{proposition}\label{proposition:symmetric}
The polymatroid $g$ has an $n$-page extension if and only if it has such a
symmetrical extension.
\end{proposition}
\begin{proof}
Let $h$ be an $n$-page book extension of $g$. For any permutation $\pi$ of
the pages define the polymatroid $\pi h$ on $M$ so that
$(\pi h)(I) = h(\sigma_\pi(I))$.
This polymatroid is also an $n$-book extension of $g$ with 
the same bijections $\phi_i$, and consequently
$$
    \frac1{n!}\: \tsum_{\pi}\, \pi h 
$$
is again an $n$-page book extension of $g$, which is symmetrical.
\end{proof}

The next theorem is a generalization of \cite[Theorem
3]{matus-csirmaz}. It will be used in proving the main result 
of this section, Theorem \ref{thm:book-on-singletons}, and it
essentially
shows that for book extensions it is enough to consider tight 
polymatroids.

\begin{theorem}\label{thm:balanced}
{\rm a)} Suppose there is an $n$-page book extension of $\langle g,N\rangle$
over $S$. Then $g\dn N$ also has an $n$-page book extension over $S$.
{\rm b)} Suppose $\langle h',M\rangle$ is an $n$-page book extension of $g'=g\dn N$
over $S$. Then there is an $n$-page book extension $g\prec^n_S h$ such that
$h\dn M=h'$.
\end{theorem}
\begin{proof}
Part a) follows by induction on the number of elements in $N$ from the
following claim: if $a\in N$ and $g$ has an $n$-page book extension 
$\langle h,M\rangle$, 
then so has $g\dn a$. So fix $a\in N$, and the extension $g\prec^n_S h$.
Let $\phi_i$ be the bijection between $PS$ and $P_iS$, and let
$t = g(N)-g(N{-}a)$. Consider first the case when $a\in S$. Define $h'$ 
on the subsets of $M$ by
$$
   h'(J) = \min \,\{\, h(J),h(aJ)-t\,\} .
$$
This is a polymatroid by Lemma \ref{lemma:minimum-poly}, and
$g\dn a\prec^n_S h'$. Indeed, the $\phi_i$ pullback of $h'$ is $g\dn a$
trivially. Furthermore, as $a\in S$, 
$h'(P_{\langle i\rangle}S) = h(P_{\langle i\rangle}S)-t$ for
any subset $\langle i\rangle$ of $\{1,2,\dots,n\}$, thus $(P_{\langle
i\rangle},P_{\langle j\rangle}\|S )\cdot h=0$ implies $(P_{\langle
i\rangle},P_{\langle j\rangle}\|S )\cdot h'=0$, that is, the pages are
totally independent over $S$ in $h'$ as well.

In the second case $a\in P$. We denote $\phi_i(a)\in P_i$ by $a_i$, 
and call it the {\em twin of $a$}.
Let $h_0=h$, and 
define for $1\le\ell\le n$ the polymatroid $h_\ell$ on $M$ as follows:
$$
    h_\ell(J)=\min \,\{\, h_{\ell-1}(J), h_{\ell-1}(a_\ell J)-t\,\}.
$$
The following holds: the $\phi_i$ 
pullback of $h_\ell$ is $g\dn a$ when $i\le\ell$, and is $g$ otherwise; 
and the pages are totally independent over $S$ in $h_\ell$. This is
true for $\ell=0$, and we prove it by induction for all $\ell\le n$ below. 
Thus $h'=h_n$ is an $n$-page extension of $g\dn a$, which completes the
induction step for part a).

Suppose the above claim for $\ell-1$; pick $i\not=\ell$ and
$J\subseteq P_i$ arbitrarily. By submodularity and by the induction 
assumption
\begin{align*}
      & h_{\ell-1}(a_\ell J)-h_{\ell-1}(J) \ge {} \\
{}\ge{} & h_{\ell-1}(a_\ell P_iS)-h_{\ell-1}(P_iS) = {} \\
{}={}   & h_{\ell-1}(a_\ell S)-h_{\ell-1}(S) = g(aS)-g(S) \ge t,
\end{align*}
which proves $h_\ell(J)=h_{\ell-1}(J)$, that is, for $i\not=\ell$ the
$\phi_i$ pullbacks of $h_\ell$ and $h_{\ell-1}$ are the same. The
$\phi_\ell$ pullback of $h_\ell$ is clearly $g\dn a$. Finally,
the independence of the pages in $h_\ell$ follows from their
independence in $h_{\ell-1}$ and from
$$
h_\ell(P_{\langle i\rangle}S) = \begin{cases}
    \,h_{\ell-1}(P_{\langle i\rangle}S) & \text{ if $\ell\notin\langle
i\rangle$,}\\
    \,h_{\ell-1}(P_{\langle i\rangle}S)-t & \text{ if $\ell\in\langle
i\rangle$} .
  \end{cases}
$$

\smallskip

For part b), let $a\in N$, $g'=g\dn a$, and suppose $g'\prec^n_S h'$. We
claim the existence of a polymatroid $h$ with $g\prec^n_S h$
such that (i) if $a\in S$, then $h\dn a=h'$; and (ii) if $a\in P$, then
$h\dn a_1\dots a_n=h'$.
In case (i) first we check $h'=h'\dn a$.
The $\phi_i$-pullback of $h'$ is $g'$, $g'(N)=g'(N-\{a\})$, thus
\begin{align*}
   0 &\le h'(M)-h'(M-\{a\}) \\
     &\le h'(P_iS)-h'(P_iS-\{a\})\\
     & = g'(PS)-g'(PS-\{a\}) = 0,
\end{align*}
establishing $h'=h'\dn a$.
Now let 
$t=g(N)-g(N-\{a\})$, then $g'=g\dn a$ means that for every $I\subseteq
N-\{a\}$, $g(I) = g'(I)$, and $g(aI) = g'(aI)+t$.
Let us define the polymatroid $h$ on the ground set $M$ by
$$
\begin{array}{r@{}l }
    h(J) &{}= h'(J), \\[5pt]
    h(aJ) &{} = h'(aJ)+t,
\end{array}
~~ J\subseteq M-\{a\} .
$$
It is clear that $g$ is the $\phi_i$-pullback of $h$ and $h\dn a=h'\dn a =
h'$. To conclude that $g\prec^n_S h$,
only the independence of pages in $h$ should be checked.
To this end let $\langle i\rangle$ and $\langle
j\rangle$ be two disjoint non-empty subsets of $\{1,2\dots, n\}$.
Then
\begin{align*}
  &  h(P_{\langle i\rangle}S)+h(P_{\langle j\rangle}S) = 
     2t + h'(P_{\langle i\rangle}S)+h'(P_{\langle j\rangle}S) ={}\\
  & ~~ = 2t + h'(P_{\langle i\rangle}P_{\langle j\rangle}S) + h'(S) 
     = h(P_{\langle i\rangle}P_{\langle j \rangle}S)+h(S) .
\end{align*}
Here we used the facts that $a\in S$, and $P_{\langle i\rangle}$ and $P_{\langle
j\rangle}$ are independent in $h'$. This concludes part (i).

In case (ii), when $a\in P$, $h'=h'\dn a_1\dots a_n$ follows as above. Setting
$t=g(N)-g(N-\{a\})$, define the polymatroid $h$ by
$$
\begin{array}{r@{}l }
    h(J) &{}= h'(J),\\
    h(a_{\langle i \rangle}J)&{}= h'(a_{\langle i\rangle}J)+
    t\cdot | \langle i\rangle |
\end{array}
~~ J\subseteq M-\{a_1,\dots,a_n\},
$$
where $|\langle i\rangle |$ is the
cardinality of the set $\langle i\rangle$.
Now $h\dn a_1\dots a_n=h'\dn a_1\dots a_n=h'$, and as $a_i\in P_i$,
the independence also holds:
\begin{align*}
   & h(P_{\langle i\rangle}S)+h(P_{\langle j\rangle}S)= \\
   &~~= \big(h'(P_{\langle i\rangle}S)+t\cdot|\langle i\rangle|\big) + 
    \big(h'(P_{\langle j\rangle}S)+t\cdot|\langle j\rangle|\big)
    = {}\\
   &~~=\big(h'(P_{\langle i\rangle}P_{\langle j\rangle}S)+
      t\cdot(|\langle i\rangle|+|\langle j\rangle|)\big) +h'(S) = {}\\
   &~~= h(P_{\langle i \rangle}P_{\langle j\rangle}S)+h(S).
\end{align*}
Claim b) follows from (i) and (ii) by induction on the number of the 
elements of $N$.
\end{proof}

A {\em co-singleton} is a subset which misses only one element.

\begin{theorem}\label{thm:book-on-singletons}
Every polymatroid $\langle g,N\rangle$ has an $n$-page extension over 
singletons and co-singletons.
\end{theorem}

\begin{proof} 
First let $S=\{a\}$ and $P=N-\{a\}$. Let $\langle f,N^* \rangle$ be the direct
sum of $n$ disjoint copies of $\langle g,N\rangle$ where 
$N^* = N_1\cups\dots\cups N_n$. Denote $N_i-\{a_i\}$ by $P_i$ where $a_i\in N_i$ is
the copy of $a$, and let $M=P_1\cups\dots\cups P_n\cups\{a\}$. Define the 
map $\phi: N^*\to M$ so that $\phi(a_i)=a$, otherwise $\phi$ is the identity.
Applying Lemma \ref{lemma:minimum-poly} to the pullback $\phi^{-1}f$ 
and $t=(n-1)g(a)$ gives the polymatroid $h$ on $M$, which will be the required
extension. The independence of $P_1$, \dots, $P_n$ over $\{a\}$ follows from the
fact that
$$
    h(aP_{\langle i\rangle}) = g(a)+|\langle i\rangle|\big(g(aP)-g(a)\big).
$$
The restrictions of $h$ to $P_i\cup\{a\}$ are clearly isomorphic to $g$.

To prove the second claim of the theorem, suppose that $S\subset N$ is a 
co-singleton and $P=N-S=\{a\}$. By Theorem
\ref{thm:balanced}, it is enough to show that $g'=g\dn N$ has a
$k$-page extension. Let $M=\{a_1,\dots,a_n\}\cups S$, and let $\phi(a_i)=a$,
$\phi\restr S = \mathrm{Id}\restr S$. 
Define $h'$ on $M$ as the pullback $\phi^{-1} g'$.
Then, $h'$ is a polymatroid; moreover, $a_1,\dots,a_n$ are totally independent 
over $S$ as $h'(a_{\langle i\rangle}S) = g'(aS)=g'(S)$. Consequently, 
$g'\prec^n_S h'$, which was to be shown.
\end{proof}


\section{Book inequalities}\label{sec:conjecture}

\newcommand\vv{\mathbf{v}}

From this section on, we concentrate on polymatroids on a four-element ground
set $N$, whose elements will be denoted by the letters $a$, $b$, $c$, and $d$.
The structure of these polymatroids with a special emphasis on entropic
representability have been studied extensively in
\cite{Chan.progress, dougherty-six, dougherty-etal, 
matroid-bounds, matus-csirmaz, matus-studeny, ww, xu-wang-sun}.
According to Theorem \ref{thm:book-on-singletons}, every polymatroid on a
four-element set
has book extensions over singletons and over three-element subsets.
Existence of book extensions over the two-element subset $S=\{a,b\}$
can be characterized in terms of linear inequalities: such
a polymatroid $g$ has 
an $n$-page book extension over $S$ if and only if it satisfies a 
certain collection of linear inequalities. As the existence of 
a book extension is invariant for permutations of the ground set which
keep the spine $S$ fixed, this characterizing set of inequalities is also
invariant under these permutations. In this case the stabilizer of the spine
$S=\{a,b\}$ is generated by two permutations, those which swap 
$a\leftrightarrow b$ and $c\leftrightarrow d$, respectively; thus, the
collection of inequalities is invariant under these swaps of variables. 
For the 2-page case {F.} Mat\'u\v s provided the following
characterization.

\def\mcite{\cite[Theorem 3]{matus-adhesive}}
\begin{theorem}[Mat\'u\v s \protect\mcite]\label{thm:matus-2book}
A polymatroid $g$ on the four-element set $abcd$ has a 2-page extension over $ab$ if and only if 
the following instances of the Zhang-Yeung inequality, and their
$a\leftrightarrow b$ and $c\leftrightarrow d$ versions, hold for $g$:
\begin{alignat*}{2}
&[abcd\,] + (a,b\|c) + (a,c\|b) + (b,c\|a) &&\ge 0,\\ 
&[bdac\,] + (a,b\|d) + (b,d\|a) + (a,d\|b) &&\ge 0.
\end{alignat*}
\end{theorem}

Using computer-aided multiobjective optimization, the characterizing
collection of linear inequalities were generated for up to 9-page book
extensions. Based on these experiments, the collection of {\em $n$-page book
inequalities} is defined below, and it
is conjectured to characterize polymatroids which have $n$-page book
extensions over $ab$. In Section \ref{sec:book-proof} these inequalities are
shown to hold for such polymatroids, thus they are information inequalities.
The sufficiency of the characterization is left as an open problem.

The description of the book inequalities is rather involved. The set of
non-negative integers $\{0,1,\dots\}$ is denoted by $\N$. Among the finite
subsets of the non-negative lattice points $\N\times\N$, the following
subsets will be of particular interest for integers $n\ge 2$:
\begin{align*}
 u_n &\eqdef \{ \langle k,0\rangle\in\N\times\N \,:\, k\le n-2 \}, \\
 v_n &\eqdef \{ \langle 0,\ell\rangle\in\N\times\N \,:\, \ell\le n-2 \}, \\
 t_n &\eqdef \{ \langle k,\ell\rangle\in\N\times\N \,:\, k+\ell\le n-2 \}.
\end{align*}
For $k,\ell\in\N$ the three-dimensional integer vector $\vv_{k,\ell}$ is
defined as
$$
    \vv_{k,\ell} = \binom{k+\ell}{k}\langle 1,k+1,\ell\rangle .
$$
For example, $\vv_{k,0}=\langle 1,k+1,0\rangle$, 
$\vv_{0,\ell}=\langle1,1,\ell\rangle $ and $\vv_{k-1,1}=\langle
k,k^2,k\rangle$, $\vv_{1,\ell-1}=\langle \ell,2\ell,\ell^2-\ell\rangle$.
For a finite subset $s$ of the lattice points let $\vv_s$ be the sum of the
vectors $\vv_{k,\ell}$ when $\langle k,\ell\rangle$ runs over $s$:
$$
    \vv_s \eqdef
        \tsum \,\{ \vv_{k,\ell} \,:\, \langle k,\ell\rangle \in s \} .
$$
This value, computed for the subsets $u_n$, $v_n$ and $t_n$, gives
\begin{align*}
\vv_{u_n} &=\tsum_{k\le n-2} \vv_{k,0} = \langle n-1,n(n-1)/2,0\rangle,\\
\vv_{v_n} &=\tsum_{\ell\le n-2} \vv_{0,\ell}=\langle n-1,n-1,(n-1)(n-2)/2\rangle ,\\
\vv_{t_n} &= \langle2^{n-1}-1,(n-1)2^{n-2},(n-2)2^{n-2}+1\rangle.
\end{align*}
The value $\vv_s$ is symmetrical in the following sense: if $\vv_s=\langle
x_s,y_s,z_s\rangle$ and $s^T$ is 
the transpose of $s$, that is
$s^T=\{\langle\ell,k\rangle: \langle k,\ell\rangle\in s\}$, 
then $\vv_{s^T}=\langle x_s,z_s+x_s,y_s-x_s\rangle$.

\begin{figure}[h!tb]
\begin{center}\begin{tikzpicture}[scale=0.6]
\begin{scope}[xshift=0cm]
\draw[->] (0,0)--(0,1.1); \draw[->] (0,0)--(1.1,0);
\foreach \x/\y in {0/0,0/0.4,0/0.8}{\draw[fill] (\x,\y) circle (2.5pt); }
\foreach \x/\y in {0.4/0,0.4/0.4,0.8/0}{\draw (\x,\y)circle(2.5pt); }
\draw (0.4,-0.5) node{$v_4$};
\end{scope}
\begin{scope}[xshift=1.5cm]
\draw[->] (0,0)--(0,1.1); \draw[->] (0,0)--(1.1,0);
\foreach \x/\y in {0/0,0/0.4,0/0.8,0.4/0}{\draw[fill] (\x,\y) circle (2.5pt); }
\foreach \x/\y in {0.4/0.4,0.8/0}{\draw (\x,\y)circle(2.5pt); }
\end{scope}
\begin{scope}[xshift=3cm]
\draw[->] (0,0)--(0,1.1); \draw[->] (0,0)--(1.1,0);
\foreach \x/\y in {0/0,0/0.4,0/0.8,0.4/0,0.4/0.4}{\draw[fill] (\x,\y) circle
(2.5pt); }
\foreach \x/\y in {0.8/0}{\draw (\x,\y)circle(2.5pt); }
\end{scope}
\begin{scope}[xshift=4.5cm]
\draw[->] (0,0)--(0,1.1); \draw[->] (0,0)--(1.1,0);
\foreach \x/\y in {0/0,0/0.4,0/0.8,0.4/0,0.4/0.4,0.8/0}{\draw[fill] (\x,\y) circle
(2.5pt); }
\draw (0.4,-0.5) node{$t_4$};
\end{scope}
\begin{scope}[xshift=6cm]
\draw[->] (0,0)--(0,1.1); \draw[->] (0,0)--(1.1,0);
\foreach \x/\y in {0/0,0/0.4,0.4/0,0.4/0.4,0.8/0}{\draw[fill] (\x,\y) circle
(2.5pt); }
\foreach \x/\y in {0/0.8}{\draw (\x,\y)circle(2.5pt); }
\end{scope}
\begin{scope}[xshift=7.5cm]
\draw[->] (0,0)--(0,1.1); \draw[->] (0,0)--(1.1,0);
\foreach \x/\y in {0/0,0/0.4,0.4/0,0.8/0}{\draw[fill] (\x,\y) circle (2.5pt); }
\foreach \x/\y in {0/0.8,0.4/0.4}{\draw (\x,\y)circle(2.5pt); }
\end{scope}
\begin{scope}[xshift=9cm]
\draw[->] (0,0)--(0,1.1); \draw[->] (0,0)--(1.1,0);
\foreach \x/\y in {0/0,0.4/0,0.8/0}{\draw[fill] (\x,\y) circle (2.5pt); }
\foreach \x/\y in {0/0.4,0/0.8,0.4/0.4}{\draw (\x,\y)circle(2.5pt); }
\draw (0.4,-0.5) node{$u_4$};
\end{scope}
\begin{scope}[xshift=10.5cm]
\draw[->] (0,0)--(0,1.1); \draw[->] (0,0)--(1.1,0);
\foreach \x/\y in {0/0,0/0.4,0.4/0,0.4/0.4}{\draw[fill] (\x,\y) circle (2.5pt); }
\foreach \x/\y in {0.8/0,0/0.8}{\draw (\x,\y)circle(2.5pt); }
\end{scope}
\begin{scope}[xshift=12cm]
\draw[->] (0,0)--(0,1.1); \draw[->] (0,0)--(1.1,0);
\foreach \x/\y in {0/0,0/0.4,0.4/0,0.8/0,0/0.8}{\draw[fill] (\x,\y) circle (2.5pt); }
\foreach \x/\y in {0.4/0.4}{\draw (\x,\y)circle(2.5pt); }
\end{scope}
\end{tikzpicture}\end{center}
\kern -15pt
\caption{Subsets in $\mathcal S_4 - \mathcal S_3$}\label{fig:S4minusS3}
\end{figure}
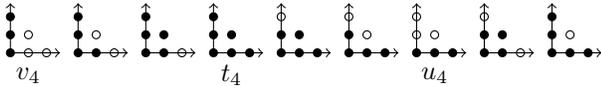
The subset $s$ of the lattice points is {\em downward closed} if from 
$\langle
k,\ell\rangle\in s$ and $0\le k'\le k$, $0\le\ell'\le \ell$, it follows that
$\langle k',\ell'\rangle\in s$. In particular, $u_n$, $v_n$ and $t_n$
are downward closed sets.
For $n\ge 2$ let us define $\Sn$ as the collection of downward closed
subsets of $t_n$:
$$
    \Sn\eqdef \big\{s\subseteq t_n \,: \text{ $s$ is downward closed }
\big\}.
$$
$\Sn$ is a subset of $\mathcal S_{n+1}$, and
$u_n$, $v_n$, $t_n$ are elements of $\Sn$. The family $\mathcal S_2$
has a single one-element subset $\{\langle0,0\rangle\}$, which is the same
as $u_2=v_2=t_2$. $\mathcal S_3$ has
three additional subsets: $u_3$, $v_3$, and $t_3$. Elements
of $\mathcal S_4-\mathcal S_3$ are depicted on Figure \ref{fig:S4minusS3},
the subsets $u_4$, $v_4$, and $t_4$ are marked.
\begin{definition}[Book inequalities]\label{def:book-ineq}
Let $n\ge 2$. The collection $\B_n$ of {\em $n$-page book inequalities}
is the following set of inequalities on polymatroids on the four-element set $abcd$:
\begin{align}\label{eq:book1}
   & x_s[abcd\,]  + (a,b\|c) + y_s\big((a,c\|b)+(b,c\|a)\big)\nonumber\\
   &~~~~+    z_s\big((a,d\|b)+(b,d\|a)\big)\ge 0,
\end{align}
where $\langle x_s,y_s,z_s\rangle=\vv_s$  and $s$ runs over the set $\Sn$;
plus the inequalities
\begin{equation}\label{eq:book2}
   \ell\,[bdac] + (a,b\|d)+\frac{\ell(\ell+1)}2\big(
            (b,d\|a)+(a,d\|b)\big) \ge0 .
\end{equation}
where $\ell=1,2,\dots,n-1$ .
\end{definition}

The collection of $n$-page book inequalities is increasing: every 
inequality in $\B_n$ is in $\B_{n+1}$ as well. $\B_2$ consists of the two
inequalities which appeared in Theorem \ref{thm:matus-2book}; this is so
as $\mathcal S_2$ has the only element $\{\langle 0,0\rangle\}$, and
$\vv_{0,0}=\langle 1,1,0\rangle$.
The sequence $B_n$ contains two previously identified infinite 
lists of entropy inequalities. Setting $s=u_n\in \Sn$, inequality
(\ref{eq:book1}) becomes
$$
   (n-1)[abcd\,] + (a,b\|c) + \frac{n(n-1)}2 \big((a,c\|b)+(b,c\|a)\big)
      \ge 0 ,
$$
which is one of the (implicit) infinite families of new entropy
inequalities from \cite[Theorem 2]{matus-infinite}. When $s$ is
$t_n\in\Sn$, then (\ref{eq:book1}) becomes
\begin{align*}
& (2^{n-1}-1)[abcd\,] + (a,b\|c) \\
&~~~~ +  (n-1)2^{n-2}\big((a,c\|b)+(b,c\|a)\big) \\
&~~~~ + \big((n-2)2^{n-2}+1\big)\big((a,d\|b)+(b,d\|a)\big) \ge 0,
\end{align*}
which is the inequality of \cite[Theorem 10]{dougherty-etal}.
Another interesting infinite family of inequalities arises from $v_n\in
\Sn$:
\begin{align*}
& (n-1)[abcd\,] + (a,b\|c) \\ 
&~~~~ +  (n-1)\big((a,c\|b)+(b,c\|a)\big) \\
&~~~~ + \frac{(n-1)(n-2)}2\,\big((a,d\|b)+(b,d\|a)\big) \ge 0,
\end{align*}
and several others can be constructed easily.
Some of the inequalities in $\B_n$ are redundant: they are consequences
of others. For example, $\B_4$ contains 12 inequalities, eight of them
come from the downward closed subsets depicted on Figure
\ref{fig:S4minusS3}. The $\langle x_s,y_s,z_s\rangle$ coefficients in
the order above are (3,3,3), (4,5,3), (6,9,5), (7,12,5), (6,11,3), (4,7,1),
(3,6,0), (5,8,3),
and (5,8,3). The inequality coming from the last two triplets is a
consequence of the others, as it is just the average of the inequalities
coming from the triplets (4,5,3), and (6,11,3).
It is not difficult to eliminate the redundant inequalities from $
\B_n$ but their description is cumbersome, so we skipped this step.
\begin{figure}
\begin{center}\hbox to \columnwidth{\hss%
\begin{tikzpicture}
\begin{scope}
\useasboundingbox(-0.2,-0.2) rectangle (7.3,7.3);
\clip(-0.1,-0.1) rectangle(7.2,7.2);
\input e4.tex
\end{scope}
\draw (0,0) node[below]{$\scriptstyle1$};
\draw (1.520,0) node[below]{$\scriptstyle3/2$};
\draw (2.599,0) node[below]{$\scriptstyle2$};
\draw (3.436,0) node[below]{$\scriptstyle5/2$};
\draw (4.120,0) node[below]{$\scriptstyle3$};
\draw (5.199,0) node[below]{$\scriptstyle4$};
\draw (6.035,0) node[below]{$\scriptstyle5$};
\draw (6.719,0) node[below]{$\scriptstyle6$};
\draw (0,0) node[left]{$\scriptstyle0$};
\draw (0,1.520) node[left]{$\scriptstyle1/2$};
\draw (0,2.599) node[left]{$\scriptstyle1$};
\draw (0,3.436) node[left]{$\scriptstyle3/2$};
\draw (0,4.120) node[left]{$\scriptstyle2$};
\draw (0,5.199) node[left]{$\scriptstyle3$};
\draw (0,6.035) node[left]{$\scriptstyle4$};
\draw (0,6.719) node[left]{$\scriptstyle5$};
\end{tikzpicture}%
\hss}\end{center}
\kern -15pt
\caption{Coefficients $\langle y_s/x_s,z_s/x_s\rangle$ on a logarithmic scale}\label{fig:conjecture-figure}
\end{figure}
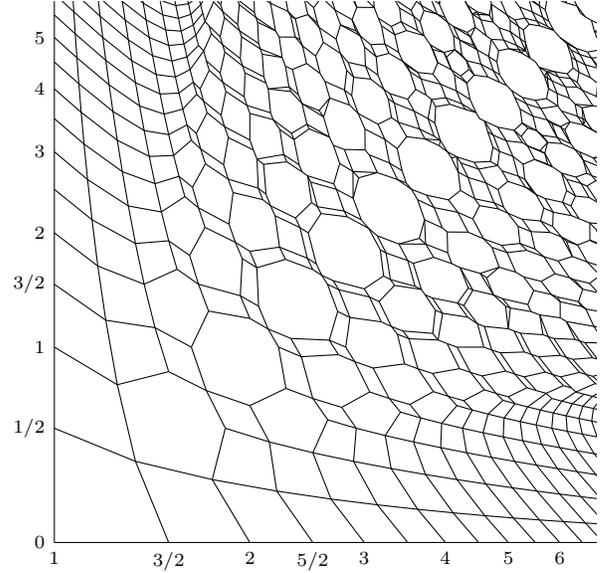
Figure \ref{fig:conjecture-figure} shows nodes
$(y_s/x_s,z_s/x_s)$ on a logarithmic scale, 
where $\langle x_s,y_s,z_s\rangle$ are
coefficients in non-redundant inequalities. Two such
nodes are connected by a straight line when the corresponding 
subsets $s\in\Sn$ differ by a
single element only. Nodes on the horizontal and vertical bounding lines come from
the sets $u_n$, $v_n$, respectively; $t_n$ gives the nodes along the diagonal.
The symmetry of the figure comes from the symmetry of $\vv_s$ observed
earlier.

\begin{conjecture}[Book conjecture]\label{conj:book}
A polymatroid $g$ on the four element set $abcd$ has an $n$-page 
book extension at $ab$ if and only if $g$ satisfies the $n$-page book 
inequalities in $\B_n$ and their versions where the variables
$a\leftrightarrow b$ and $c\leftrightarrow d$ are swapped.
\end{conjecture}

The condition of this conjecture is necessary; this will be proved 
in the next section as Theorem \ref{thm:book}. Sufficiency has been
checked by a computer program for $n\le 9$. The technique used can 
be outlined as follows.
The ground set of the $n$-page book extension of $abcd$ has $2+2n$ elements,
thus the polymatroid is an element of the $2^{2+2n}-1$-dimensional Euclidean
space. The region of polymatroids is a convex polyhedral cone bounded by
half-planes corresponding to submodularity and monotonicity. The collection
of $n$-page book extensions is a sub-cone $\mathcal P$ cut out by the 
requirement that all pullbacks are isomorphic, and the pages are 
independent over the spine. These requirements can also be expressed as 
linear constraints, thus $\mathcal P$ is also polyhedral. A polymatroid on
$abcd$ has an $n$-page book extension if and only if it is in the projection
of $\mathcal P$ to the 15-dimensional subspace corresponding to the
non-empty subsets of $abcd$. The characterizing inequalities are just the
equations of the facets of the projection. Finding these facets is the
subject of multiobjective optimization. To be applicable in practice, the
problem dimension should be reduced significantly. This reduction comes from
several sources. By Proposition \ref{proposition:symmetric}, we can 
assume the book extension be symmetric, this alone drops the dimension of
$\mathcal P$ significantly from $2^{2+2n}-1$ to around $(n+1)^3$. Further
reduction is achieved from the independence of pages, from the
$a\leftrightarrow b$ and $c\leftrightarrow d$ symmetries, from the 
sufficiency of considering tight polymatroids only, and by cutting
$\mathcal P$ into several well-chosen pieces. Table \ref{table:sufficiency} 
\begin{table}[htb]
\begin{center}\begin{tabular}{|c|r@{\,$\times$\,}l|r|}
\hline
\rule{0pt}{2.2ex}$n$&\multicolumn{2}{c|}{Size}& Time \\
\hline
\rule{0pt}{2.4ex}4 &  105 & 692 &  1\\
5 & 168 & 1312 & 55 \\
6 & 252 & 2189 & 9:38 \\
7 & 360 & 3387 & 2:18:45 \\
8 & 495 & 4942 & 6:55:40 \\
9 & 660 & 6932 & 12:53:40 \\
\hline
\end{tabular}\end{center}
\caption{Checking sufficiency: problem size and running
time}\label{table:sufficiency}
\end{table}
shows, as a function of $n$,
the size of the reduced problem: its dimension and the number of linear 
constraints in that dimension. The last column contains the
running time in hours, minutes and seconds required to generate the
facets of the projection on a stand-alone workstation running a highly
optimized algorithm.

We conjecture that the book inequalities do give a sufficient condition for
the existence of an $n$-page book extension.


\section{Necessity of book inequalities}\label{sec:book-proof}

The aim of this section is to prove that the condition in the book conjecture
is necessary.
\begin{theorem}\label{thm:book}
Suppose the polymatroid $g$ on the four element set $abcd$ has an $n$-page
book extension $h$ at $ab$. Then $g$ satisfies all inequalities in $\B_n$
and their versions where the variables $a\leftrightarrow b$ and 
$c\leftrightarrow d$ are swapped.
\end{theorem}
\begin{proof}
As it was remarked earlier, it is enough to show that $g$ satisfies the inequalities
in $\B_n$ as the symmetric versions follow by applying $\B_n$ to the permuted
instances of $g$.

As $h$ is an $n$-page book extension of $g$, the ground set of $h$ is the 
disjoint union $M=P_1\cups\dots\cups P_n\cups \{a,b\}$. We let
$P_i=\{c_i,d_i\}$, where $c_i$, $d_i$ are the twins of
$c$, $d$, respectively. For non-negative integers $k$, $\ell$ and $m$ where
$k+\ell+m\le n$, let $c^kd^\ell(cd)^m$ denote the following subset
of $M$:
\begin{alignat*}{3}
    c^kd^\ell(cd)^m \eqdef\; 
   \{ &c_1,\dots,c_k, \, &&c_{k+\ell+1},\dots,c_{k+\ell+m}, \\
      &d_{k+1},\dots,d_{k+\ell},\:&&d_{k+\ell+1},\dots,d_{k+\ell+m} \},
\end{alignat*}
that is, we pick $c_i$ from the first $k$ pages, $d_i$ from the next $\ell$
pages, and both $c_i$ and $d_i$ from the following $m$ pages. When any of $k$,
$\ell$, or $m$ is zero, we leave out the corresponding term from the notation.
According to Proposition \ref{proposition:symmetric}, $h$ can be assumed to
be symmetric, that is, the value of $h(I)$ depends only
on whether $a$ and $b$ are in $I$, and in how many pages $I$ intersects
$P_i$ in the empty set, in $c_i$, in $d_i$, or in $c_id_i$. Consequently 
$h(I)$ is equal to one of the values $h(X)$, $h(aX)$, $h(bX)$ or $h(abX)$, 
where $X=c^kd^\ell(cd)^m$ for some triplet $k,\ell,m$.

To simplify the notation, in the rest of this section we omit the symbols $g$ and $h$
before the subsets of $abcd$ and $M$; any subset also 
denotes the value of the corresponding polymatroid. As $g$ and $h$ 
agree on subsets of $ab$, this convention is unambiguous.

First we prove some easy propositions.
\begin{claim}\label{claim:ab}
{\rm a)} If $k+\ell\le n$, then
$$
abc^kd^\ell = ab+k\cdot(abc-ab) + \ell\cdot(abd-ab);
$$
{\rm b)} if $k+\ell\le n-1$, then 
$$abc^kd^\ell(cd)^1 = ab(cd)^1 +k\cdot(abc-ab) + \ell\cdot(abd-ab).
$$
\end{claim}
\begin{proof}
By induction on $k+\ell$.
Both statements are true when $k=\ell=0$. Assume $k+\ell<n$. 
As $c_{n}$ is independent of $c^kd^\ell$ over $ab$, that is,
$(c_n,c^kd^\ell \| ab) = 0$, and as $abc_n=abc$, $abc_nc^kd^\ell =
abc^{k+1}d^\ell$, we know that
$$
     abc^{k+1}d^\ell - abc^kd^\ell = abc-ab,
$$
and similarly for the other three cases. This concludes the
induction step.
\end{proof}

\begin{claim}\label{claim:firstcase}
If $k+\ell\le n$, then
\begin{align*}
k\cdot(ac-a) + \ell\cdot(ad-a) + a &\ge ac^kd^\ell, \\
k\cdot(bc-b) + \ell\cdot(bd-b) + b &\ge bc^kd^\ell .
\end{align*}
\end{claim}
\begin{proof}
The claims are true with equality when $k=\ell=0$. As 
$(c_n,c^kd^\ell \|a)\ge 0$ and $ac_n=ac$, $c_nc^kd^\ell=c^{k+1}d^\ell$,
we know that $ac - a \ge ac^{k+1}d^\ell - ac^kd^\ell$. Using this fact
and three other similar inequalities we arrive at the claim by induction on 
$k+\ell$.
\end{proof}

\begin{claim}\label{claim:second}
If $k+\ell<n$, then
\begin{align*}
    c^kd^\ell(cd)^1 &\ge cd+k(abc-ab) + \ell(abd-ab),\\
    bd^\ell(cd)^1 &\ge bcd + \ell(abd-ab), \\
    acd^\ell &\ge ac + \ell(abd-ab).
\end{align*}
\end{claim}
\begin{proof}
By submodularity, $c^kd^\ell(cd)^1 - (cd)^1 \ge abc^kd^\ell(cd)^1 - ab(cd)^1$. This,
and part b) of Claim \ref{claim:ab} give the first inequality. The other
inequalities can be proved in a similar way.
\end{proof}

The next lemma describes the crucial inequality that allows us to prove
that $g$ satisfies the inequalities in $\B_n$. The symbols $\C$, $\D$ will be 
used to denote the following entropy expressions:
\begin{align*}
\C &= (a,c\|b) + (b,c\|a), \\
\D &= (a,d\|b) + (b,d\|a) .
\end{align*}

\begin{lemma}\label{lemma:crucial}
For non-negative integers $k$ and $\ell$ where $k+\ell<n$,
\begin{align}
& [abcd\,] + k\C+\ell\D + (a,b\|c^kd^\ell) \ge{}\label{eq:crucial1}\\
& ~~~~ (a,b\|c^{k+1}d^\ell) + (a,b\|c^kd^{\ell+1}) + (c_n,d_n|c^kd^\ell),
      \nonumber
\end{align}
and
\begin{align}
&  [bdac\,] + \ell\D + (a,b\|d^\ell) \ge {}\label{eq:crucial2}\\
&~~~~ (a,b\|d^{\ell+1})+(a,c\|d^\ell)+(b,d_n\|c_nd^\ell). \nonumber
\end{align}
\end{lemma}
Before proving this lemma, let us see how it implies Theorem \ref{thm:book}.
Denote the inequality (\ref{eq:crucial1}) by $\I(k,\ell)$, and the
inequality (\ref{eq:crucial2}) by $\J(\ell)$.
First let $s\in \Sn$ and $\vv_s=\langle x_s,y_s,z_s\rangle$, we want to show 
inequality (\ref{eq:book1}), which can be written as
\begin{equation}\label{eq:book-first}
    x_s[abcd\,] + (a,b\|c) + y_s\C + z_s\D \ge 0 .
\end{equation}
Consider the following combination of the inequalities
in (\ref{eq:crucial1}) over the elements of the downward closed set $s$:
$$
    \sum_{\langle k,\ell\rangle\in s}\:  \binom{k+\ell}{k} \I(k+1,\ell) .
$$
On the left hand side of the $\ge$ sign we have $x_s$, $y_s$, and $z_s$-many 
instances of $[abcd\,]$, $\C$, and $\D$, respectively; and we also have
$(a,b\|c)$ from $\I(1,0)$. If $k+\ell\ge 1$, then
$(a,b\|c^{k+1}d^\ell)$ occurs $\binom{k+\ell}{k}$ many times on the left
hand side, and, as $s$ is downward closed, $\binom{k-1+\ell}{k-1} +
\binom{k+\ell-1}{k}$ times on the right hand side, thus they cancel out. All
remaining items on the right hand side are non-negative, which proves
inequality (\ref{eq:book-first}).

The book inequality (\ref{eq:book2}) requires us to show
\begin{equation}\label{eq:booksecond}
    \ell[bdac\,] + (a,b\|d) + \frac{\ell(\ell+1)}2\,\D \ge 0
\end{equation}
for every $\ell<n$. Summing up the inequalities $\J(1)$, $\J(2)$, \dots,
$J(\ell)$
we get an inequality where the left hand side equals that of 
(\ref{eq:booksecond}), and where terms on the right hand side
are non-negative.
\end{proof}

\begin{proof}[Proof of Lemma \ref{lemma:crucial}]
To arrive at inequality (\ref{eq:crucial1}) sum up the inequalities in the
list below, and rearrange. The last column indicates why the inequality
holds: SM stands for submodularity, and numbers refer to the
corresponding Claim:
\newcommand\lbel[1]{\hbox to 1em{\hss\normalfont\footnotesize #1}}
$$\begin{array}{r@{\,\;}c@{\,\;}l@{\qquad}l}
  ac - a \,+\, ac^kd^\ell &\ge& ac^{k+1}d^\ell, 
     &\lbel{SM}\\[4pt]
  bd - b \,+\, bc^kd^\ell &\ge& bc^kd^{\ell+1}, 
     &\lbel{SM}\\[4pt]
 -abd + ab -abc^kd^\ell &=& -abc^kd^{\ell+1}, 
     &\lbel{\ref{claim:ab}} \\[4pt]
 - abc  - k(abc-ab) - \ell(abd-ab) &=&  - abc^{k+1}d^\ell, 
      &\lbel{\ref{claim:ab}} \\[4pt]
  ad + k(ac-a)+\ell(ad-a) &\ge& ac^kd^{\ell+1},
      &\lbel{\ref{claim:firstcase}} \\[4pt]
  bc + k(bc-b)+\ell(bd-b) &\ge& bc^{k+1}d^\ell,
      &\lbel{\ref{claim:firstcase}} \\[4pt]
 -cd -k(abc-ab)-\ell(abd-ab) &\ge& -c^kd^\ell(cd)^1 .
      &\lbel{\ref{claim:second}}
\end{array}$$
Similarly, inequality (\ref{eq:crucial2}) follows from the sum of the
inequalities in the list below:
$$\begin{array}{r@{\,\;}c@{\,\;}l@{\qquad}l}
  cd - d &\ge& d^\ell(cd)^1 - d^{\ell+1}, 
     &\lbel{SM}\\[4pt]
  bd - b \,+\, bd^\ell &\ge& bd^{\ell+1}, 
     &\lbel{SM}\\[4pt]
 -abd + ab -abd^\ell &=& -abd^{\ell+1}, 
     &\lbel{\ref{claim:ab}} \\[4pt]
  ad + \ell(ad-a) &\ge& ad^{\ell+1},
      &\lbel{\ref{claim:firstcase}} \\[4pt]
  bc + \ell(bd-b) &\ge& bcd^\ell,
      &\lbel{\ref{claim:firstcase}} \\[4pt]
 -ac -\ell(abd-ab) &\ge& -acd^\ell ,
      &\lbel{\ref{claim:second}} \\[4pt]
 -bcd -\ell(abd-ab) &\ge& -bd^\ell(cd)^1 .
      &\lbel{\ref{claim:second}}
\end{array}$$
\end{proof}

\section*{Acknowledgment}
The author would like to acknowledge the numerous insightful, fruitful, and
enjoyable discussions with Frantisek Mat\'u\v s on the entropy function,
matroids, and on the ultimate question of everything.

A preliminary version of this paper was presented at the First Workshop on 
Entropy and Information Inequalities, held in Hong Kong, April 15--17, 2013. 
The author
would like to express his gratitude to the organizers for their hospitality
and to the participants for the fruitful discussions.

\begin{IEEEbiographynophoto}{L\'aszl\'o Csirmaz}
has been with Central European University, Budapest, since 1996. Before that
he worked as a researcher at the R\'enyi Institute of Mathematics, Budapest.
His main research interests include secret sharing, Shannon theory, and
combinatorial games.
\end{IEEEbiographynophoto}


\begin{thebibliography}{99}

\bibitem{Chan.progress}
T.~H.~Chan (2011),
    Recent progresses in characterising information inequalities.
    \emph{Entropy} {\bf 13}(2) 379--401.

\bibitem{dougherty-six}
R.~Dougherty, C.~Freiling, and K.~Zeger (2006)
    Six New Non-Shannon Information Inequalities.
    \emph{Proceedings IEEE ISIT 2006},
    Seattle, Washington, 233--236.

\bibitem{dougherty-etal}
R.~Dougherty, C.~Freiling, K.~Zeger (2011),
{\em Non-Shannon information inequalities in four random variables}
\newblock ArXiv:1104.3602 (April 2011),  accessed Dec. 2013.

\bibitem{ingleton}
A.~W.~Ingleton (1971)
Conditions for representability and trasversality of matroids.
\newblock {\em Proc.\ Fr.\ Br.\ Conf.\ 1970}, Springer Lecture Notes {\bf 211},
    Springer-Verlag, Berlin, 62--67.

\bibitem{kaced}
T.~Kaced (2013),
Equivalence of two proof techniques for non-Shannon type inequalities.
\newblock arXiv:1302.2994 (February 2013), accessed Dec. 2013.

\bibitem{matroid-bounds}
C.~Li, J.~McLaren Walsh, S.~Weber (2013),
Matroid bounds on the region of entropic vectors. In:
\newblock 51th Annual Allerton Conference on Communication, 
Control and Computing, Oct. 2013.

\bibitem{lovasz}
L.~Lovasz (1982)
Submodular functions and convexity. In:
{\em Mathematical Programming -- the state of art} (A.~Bachen, M.~Gr\"otchel
and B.~Korte, eds),
\newblock Springer Verlag, pp. 234--257.

\bibitem{MMRV}
 K.~Makarychev, Yu.~Makarychev, A.~Romashchenko and N.~Vereshchagin (2002),
\newblock A new class of non-Shannon-type inequalities for entropies.
\emph{Communications in Information and Systems} {\bf 2} 147--166.

\bibitem{matus-adhesive}
F.~Matus (2007),
Adhesivity of polymatroids,
\newblock {\em Discrete Mathematics} vol 307 (2007) pp. 2464--2477.

\bibitem{matus}
F.~Matus (2007),
Two constructions on limits of entropy functions,
\newblock {\em IEEE Trans. Inform. Theory}, Vols 53(1) (2007) pp. 320--330.

\bibitem{matus-infinite}
F.~Matus (2007),
Infinitely many information inequalities,
\newblock {\em Proceedings ISIT}, June 24--29m 2007, Nice, France, pp.
41--47.

\bibitem{matus-csirmaz}
F.~Matus, L.~Csirmaz (2013),
Entropy region and convolution,
\newblock arXiv:1310.5957 (October 2013) accessed Dec. 2013.

\bibitem{matus-studeny}
F.~Matus and M.~Studeny (1995),
Conditional independencies among four random variables I,
\newblock {\em Combinatorics, Probability and Computing}, no 4, (1995) pp.
269-278.

\bibitem{ww}
J.~MacLaren Walsh, S.~Weber (2010),
Relationships among bounds for the region of entropic vectors in four
variables,
\newblock in {\em 2010 Allerton Conference on Communication, Control, and
Computing}.

\bibitem{xu-wang-sun}
W.~Xu, J.~Wang, J.~Sun (2008),
A projection method for derivation of non-Shannon-type information
inequalities,
\newblock in {\em Proc. IEEE International Symposium on Information Theory
(ISIT)}, (2008), pp. 2116--2120.

\bibitem{zhang-yeung}
Z.~Zhang, R.~W.~Yeung (1998),
On characterization of entropy function via information inequalities,
\newblock {\em Proc IEEE Trans. Inform. Theory},
vol 44(4) (1998) pp. 1440--1452.

\end{thebibliography}
\end{document}